\documentclass[comsoc,conference,letterpaper]{IEEEtran}
\IEEEoverridecommandlockouts
\usepackage{amssymb,amsfonts,amsthm}
\usepackage[utf8]{inputenc} 
\usepackage{graphicx}
\usepackage[T1]{fontenc}
\usepackage{url}
\usepackage{ifthen}
\usepackage{cite}
\usepackage{xcolor}
\usepackage[cmex10]{amsmath} 
\newtheorem{theorem}{Theorem} 
\newtheorem{definition}{Definition}
\newtheorem{lemma}{Lemma}
\newtheorem{example}{Example}
\newtheorem{corollary}{Corollary}
\newcommand{\bs}[1]{\ensuremath{\boldsymbol{#1}}}

\begin{document}
\title{Non-binary LDPC codes for Data Storage \thanks{ This work was supported in part by the Estonian Research Council through the grant PRG2531.}
  }

 \author{
  \IEEEauthorblockN{Irina Bocharova, Boris Kudryashov, Henk D.L. Hollmann, and Vitaly Skachek }
 \IEEEauthorblockA{Institute of Computer Science\\
                    University of Tartu\\
                    Tartu, Estonia\\
                    Email: \{irinaboc, boriskud,  henk.hollmann,vitaly.skachek\}@ut.ee} }

\maketitle  


\begin{abstract}
In modern data storage systems, non-binary LDPC codes for recovering from disk failures are increasingly considered strong competitors to MDS codes such as Reed-Solomon codes. Since disk failures can be modeled as erasures, we analyze non-binary LDPC codes over a $q$-ary field in the $q$-ary erasure channel, relative to MDS codes. Our focus is on non-binary LDPC codes whose parity-check matrix is obtained by replacing the non-zero entries of a binary base matrix by elements of a $q$-ary finite field, for any alphabet size $q$. For such LDPC codes, we introduce the notion of ultimate distance, which upper-bounds their minimum distance. We derive a random-coding bound on the number of non-correctable erasure patterns for the Gallager ensemble of regular non-binary LDPC codes under maximum-likelihood decoding. An algorithm for finding the ultimate distance is presented. A low-complexity algorithm for searching for the minimum distance of the non-binary LDPC code is proposed. Finally, we construct examples of non-binary LDPC codes achieving the ultimate distance.
\end{abstract}

\section{Introduction}
Codes over nonbinary alphabets are widely used in both large-scale data storage systems and distributed storage systems 
\cite{shen2025survey,ramkumar2022codes,park2017ldpc}.
Typically, Maximum Distance Separable (MDS) codes, in particular Reed-Solomon (RS) codes, are used due to their optimal distance properties and practical decoding algorithms. 

In most standards for data communications, algebraic codes 
are replaced by Low-Density-Parity-Check (LDPC) codes,
which are preferred due to their low decoding complexity. In distributed data storage, good locality of LDPC codes, that is, the ability to recover a disk failure by using only a small subset of neighboring disks, makes them strong competitors to MDS codes, whose localization properties are poor.
Replacing MDS codes by LDPC codes in data storage applications is considered in \cite{gaidioz2007exploring,park2017ldpc,yongmei2015large}.
The application of non-binary (NB) LDPC codes to flash memory is studied in \cite{hareedy2018combinatorial}.

In \cite{bocharova2026energy}, we explored a scenario where 
$q$-ary codes address both disk failure recovery and energy efficiency. The main idea here is to exploit the redundancy of erasure-correcting codes (ECC) to minimize the energy consumption of the disk drives: if information from an idle (switched-off) disk is requested, then it can be recovered, similarly to erasure correction, without switching on the disk by exploiting the redundancy of the employed ECC. 
In that paper, a simplified model of energy-efficient data storage was formulated and studied. 

In \cite{bocharova2026energy}, it is assumed that the system manager measures the activity of disk drives and classifies them in
the range from hot to cold. Manager’s
decision about changing the state of a disk depends on its temperature and the possibility of recovering its information by using currently
active disks. Effectively, the energy-saving problem is thus reduced to the erasure correction of idle hard disks.

Unlike for data transmission over a $q$-ary erasure channel (QEC), storage applications cannot easily utilize very long LDPC codes due to complexity and response-time limitations.  Short binary LDPC codes cannot be considered competitors of $q$-ary RS codes. The $q$-ary LDPC codes outperform their binary counterparts at the cost of higher complexity decoding. Moreover, to compete with MDS codes, it is not sufficient to rely solely on belief propagation (BP) decoding to achieve the required probability of uncorrectable erasure combinations. Instead, we apply a hybrid approach, combining BP and maximum-likelihood (ML) decoding algorithms: first, the BP decoder is used to correct erasures, and then the remaining erasures are corrected by the ML decoder. Such a decoder, on average, is simpler than a decoder for an MDS code, 
but, on the other hand, increasing the code length increases the number of computations required to recover the data. This imposes additional restrictions on the code length. 

Thus, we conclude that short-length NB  LDPC codes of length below 100 symbols are important for certain applications related to storage and distributed storage systems.

For example, the analysis and simulations in \cite{bocharova2026energy} show that although NB LDPC codes have smaller minimum distances than MDS codes, they provide the highest energy saving for scenarios under consideration. The other advantage of these codes is the low complexity of recovering the disk information. The obtained results motivate further research on the construction of NB LDPC codes specifically tailored for data storage applications.
 Our goal is to optimize such codes to achieve a better tradeoff between erasure correcting performance  and decoding complexity than that of  MDS codes of similar length.

In this paper, we analyze performances of NB LDPC codes in data storage scenarios. In Section II, we reduce the disk failure recovery problem to erasure correction in the QEC. In Section III, we introduce the notion of {\em ultimate distance} for NB LDPC codes. We prove that for any alphabet size, the minimum distance of the NB LDPC code is upper-bounded by its ultimate distance. Algorithms for finding the ultimate and minimum distances of the NB LDPC code are presented in the same section. In Section IV, we derive a random coding bound on the number of non-correctable erasure patterns for the regular NB LDPC codes, and in Section V, we construct and discuss some code examples.

\section{Preliminaries}
Assume that an $[n,k]_q$ linear code of length $n$ and dimension $k$ over the field GF($q$), $q=2^m$, is used to protect data from disk failure, where codeword symbols are stored across distinct disks.
Let $H=\{h_{ij}\}$ be the $r\times n$ parity-check matrix of the code,  where $r=n-k$ is  the number of redundant symbols. If $H$ is {\em sparse\/}, we refer to the code as an {\em LDPC-code\/} \cite{Gal63}.
Here, we consider $(J,K)$-regular LDPC codes, where every row and column of~$H$ contains $J$ and $K$ nonzero elements, respectively.  The $r \times n$ binary base matrix $B=\{b_{ij}\}$ is defined by~$b_{ij}=1 $ if $h_{ij}\neq 0$ and $b_{ij}=0$ otherwise.
The binary LDPC code determined by $B$ is called the {\em base\/} (or {\em parent\/}) code, with minimum distance $d_{\rm b}$ and Tanner graph girth (see \cite{Tan81}) $g_b$. 

The parity-check matrix $H$ is constructed as
$H=\{\mu_{ij}b_{ij}\}$, 
$i=1,\dots,r$, 
$j=1,\dots,n$, 
where the nonzero coefficients 
$\mu_{ij} \in {\rm GF}($q$)\setminus \{0\}$  
are called the {\it labels}, and the matrix 
$\mathcal M=\{\mu_{ij}\}$, $i=1,\dots,r$, $j=1,\dots,n$
the {\it labeling matrix}.
We denote the minimum distance of the obtained NB LDPC code by $d_{\rm q}$. For the parity-check matrix, we use the notations $H$, $H(B)$, and $H(\mathcal M,B)$ interchangeably.

Let $\bs \xi \in \{0,1\}^n$ represent a binary length~$n$ sequence of disk states,  where non-zeros in $\bs \xi$ correspond to failed or idle drives (i.e., to erased codeword positions).
This sequence is uniquely described by 
the set $I\subseteq \{1,\ldots,n\}$ of non-zero positions, where
$|I|=\nu$ is the number of erasures. 
%
%
For a given $\nu$ and $I=\{i_1,i_2,...,i_\nu\}$,  let 
$\bs c_I=(c_{i_1}, ..., c_{i_{\nu}} )$ be the vector of unknowns at erased positions in a codeword $\bs c$, and let $\bs c_{I^{\rm c}}$ be the vector of unerased values, where $I^{\rm c}\!=\!\{1,2,...,n\! \}\setminus I$. 
Let $H_I$ denote the submatrix of~$H$ consisting of columns indexed by $I$.

The condition $\bs c H^{\rm T}=\bs 0$ implies
$\bs c_I H_I^{\rm T} = \bs c_{I^{\rm c}} H_{I^{\rm c}}^{\rm T}\triangleq {\bs s}$,
where $\bs s$ is the syndrome vector, which is known to the decoder. 
The ML decoder solves the system of linear equations
\begin{equation} \label{main1}
\bs c_I   H_I^{\rm T}=\bs s.
\end{equation} 
If the solution not unique, that is, if
$\mbox{rank}\left( H_{I}\right)<|I|$,
then the corresponding set of erasures cannot be (uniquely) corrected; otherwise, it is correctable. 
Let 
\begin{equation*} 
\mathcal B_\nu\!=\!\left\{ I: |I|=\nu, \mbox{rank}\left( H_{I}\right)<\nu \right\}
\end{equation*}
denote the set of uncorrectable patterns of $\nu$ erasures, and let $ \; b_\nu=|\mathcal B_\nu|$ count their number. Then the probability of block error in the QEC with erasure probability $\delta$ (representing either disk failure or idle status) is
\begin{equation}P_{\rm b}(\delta)=
\sum_{\nu=1}^n b_{\nu} \delta^\nu(1-\delta)^{n-\nu}, \label{block_prob}
\end{equation}
will now be used as a performance measure.
The equivalent performance measure, which depends on code only, is the conditional probability of block error given the number of erasures
\begin{equation}
P_{\text {block}}(\nu)=\frac{b_\nu}{\binom{n}{\nu}}.
\end{equation}

\section{Ultimate Distance}
In this section we introduce the notion of {\em ultimate distance} for a $q$-ary LDPC code.
\begin{definition}
For a given base matrix $B$, the {\em ultimate distance\/} $d_{\rm u}(B)$ is defined as   
\[
d_{\rm u}(B)=\max_q \max_{\mathcal M} d_{\rm q} \left(H(\mathcal M, B)\right),
\]
where $d_{\rm q} \left(H(\mathcal M, B)\right)$ is the minimum distance of the $q$-ary LDPC code with labeling matrix $\mathcal M$ and base matrix $B$.
\end{definition} 
In other words, the ultimate distance $d_{\rm u}(B)$ serves as an upper bound on the minimum distance of the corresponding $q$-ary code for large $q$. 
The ultimate distance is an important code property in applications (such as, for example, energy-efficient coding) where the alphabet size~$q$ is a design parameter. 

\begin{definition} \label{active}
For an erasure pattern $I$, a row of the corresponding base submatrix $B_I$ is called {\em active\/} if it is non-zero. The number of active rows in~$B_I$ is denoted by~$a(I)$. \end{definition}

\begin{definition} \label{proper}
A binary square matrix $A$ of order $s$ 
is {\em proper} if there exists a permutation $P=p_1, \ldots, p_s$ of its rows such that $A(p_t,t)=1$ for $t=1, \ldots, s$, otherwise, it is {\em improper}. 
Similarly, a $t\times \ell$ binary matrix $A$ is proper if all its submatrices of order $s$, $s=\min\{t,\ell\}$ are {\em proper}, otherwise $A$ is {\em improper}.  
\end{definition}

Two parity-check matrices obtained from each other by a column and row permutation determine {\em equivalent codes}. Such matrices we call {\em equivalent}.

A proper matrix is permutationally equivalent to the matrix whose main diagonal does not contain zero entries. 
For example, a binary matrix  $B$ can be a base matrix for an MDS code iff it is proper.

\begin{example} The following square matrix is improper
\[A=
\begin{pmatrix}
0 & 1 & 1 & 0 \\
0 & 1 & 1 & 0 \\
1 & 0 & 1 & 1 \\
0 & 1 & 1 & 0
\end{pmatrix}.
\]
This matrix is equivalent to
\[
\tilde{A}=
\begin{pmatrix}
 1&1&0&0\\
 1&1&0&0\\
 1&1&0&0 \\
 0&1&1&1    
\end{pmatrix}.
\]
\end{example}

The notion of a proper matrix has the following interpretation in terms of graph theory. A binary matrix $A$ can be considered as a biadjacency matrix of 
a bipartite Tanner graph~\cite{tanner1981recursive}. 
The set of rows corresponds to the set $\mathcal{C}$ of check nodes. The set of columns corresponds to the set $\mathcal V$ of variable nodes, and nonzero elements of $A$ correspond to edges. 

In these notations, a proper square matrix $A$ of order $s$ determines a graph containing a {\em complete matching} from $\mathcal C$ to $\mathcal V$ (see, e.g., \cite[Chapter 5, Th. 5.1]{LW01}). Each permutation  $P$ in Definition \ref{proper} determines a complete matching. 
The proper rectangular $A$ of size $s\times t, s\le t$
corresponds to such a Tanner graph that any subset of $s$ variable nodes is a {\em complete matching.}      

\begin{lemma} \label{lemma1}
For a binary square matrix $A$ of order $s$, for $q>2$ there exists a labeling $\mathcal M$ by elements of {\rm GF}($q$),  such that the matrix $H(\mathcal M,A)$ is non-singular, iff $A$ is proper. 
\end{lemma}
The labeling $\mathcal M$ in Lemma \ref{lemma1} is called {\em proper}.
\begin{proof}
For a square matrix $A$, the necessity can be proved by applying the Leibniz formula \cite{strang2022introduction} to the determinant of the improper matrix. Since all terms in the determinant decomposition of the improper matrix are zero, then for any labeling $H(\mathcal M,A)$ is singular. The rectangular $s\times t$ improper matrix $A$ has at least one improper submatrix of order $s$, that is, for any labeling, there exists a set of $s$ linearly dependent columns in $H(\mathcal M,A)$.    

The proof of sufficiency for a square matrix can be done by induction. 
If $A$ is proper, then there exists a sequence of proper square submatrices $A_{\nu}$ of order $\nu=1,2,...,s$.  It is evident that the statement is true for $A_1$. Assume that the statement is true for the matrix of order $A_{\nu-1}$ from this sequence, that is, $\det(H(A_{\nu-1}))\ne 0$, and by column and row permutation $A_{\nu-1}$ is reduced to the equivalent form with a nonzero main diagonal. 

W.l.g we assume that in $A_{\nu}$ the submatrix consisting of the first $\nu-1$ columns and rows is $A_{\nu-1}$. Let $x$ be the label for an element in position $(\nu,\nu)$. Label the other nonzero elements of the last row and last column with arbitrary elements of the field.  
Then the determinant of $H(A_{\nu})$, is a linear form $\Delta(x)=\alpha x +\beta$ of argument $x$ for some $\alpha, \beta \in {\rm GF}(q)$, $\alpha\neq 0$. Zero determinant can be avoided by choosing $x\neq x_0=-\beta/\alpha$. In the case of $q=3$, we choose $x=2$ if $x_0=1$ and vice versa. 
\end{proof}

\begin{corollary} \label{col1}
Given a properly labeled square matrix $A_\nu$ of order $\nu$. Assume that some proper square submatrix 
$A_{\nu-1}$ of $A_\nu$ is properly labeled and all nonzero elements outside  $A_{\nu-1}$
are labeled by random nonzero elements of GF($q$). The probability that $A_{\nu-1}$ is properly labeled is at least 
$p_q=(q-2)/(q-1)$. 
\end{corollary}
\begin{proof} 
Consider the linear form $\Delta(x)=\alpha x +\beta$ in the proof of Lemma~\ref{lemma1}, where $\alpha$ and 
$\beta$ are random elements of GF$(q)$, $\alpha \ne 0$.  For random nonzero $x$, the random variable $\Delta(x)$ takes $q-1$ 
different values  from GF($q$) with equal probabilities $1/(q-1)$. At most one of these $(q-1)$ values is zero, and the other $q-2$ values are nonzero.  
\end{proof}
The next step is to extend the result of Lemma~\ref{lemma1} to rectangular matrices. To verify that the achievable minimum distance of the code is $d$, we analyze $s=d-1$-column submatrices of $B$ and check that each such submatrix has full rank after labeling. First we introduce the following definition. 

\begin{definition} \label{weak}
The binary matrix $A$ of size $r\times s$, $r\ge s$, is called weakly proper if at least one order $s$ square submatrix of $A$ is proper. The labeling $\mathcal M$ that makes columns of $A$ linearly independent is called a weakly proper labeling.
\end{definition}

\begin{lemma} \label{lemma2}
Let the binary matrix $A$ of size $r\times s$ be weakly proper, and its submatrix 
of size $r\times (s-1)$ is weakly properly labeled. Then, the random labeling of the $s$-th column  
provides a weakly proper labeling for $A$ with the probability at least $p_q=(q-2)/(q-1)$. 
\end{lemma}
\begin{proof}
W.l.g., assume that the set of the first $s$ rows of $A$ is a proper matrix and its first $s-1$ columns after labeling are linearly independent. From Corollary \ref{lemma1} follows that 
random labeling of the first $s$ elements of the $s$-th column makes the order $s$ submatrix of 
$A$ non-degenerate with probability $p_q$. This guarantees the rank $s$ for the entire matrix with probability~$p_q$.
\end{proof}

\begin{example} 
Consider three base matrices shown in Fig.~\ref{fig:Tanner}.
By Definition \ref{proper}, $B_1$ is improper since the submatrix $B_1'$  consisting of the last two columns is improper. The other two matrices, $B_2$ and $B_3$ are proper. 
In terms of Definition~\ref{weak}, the submatrix $B_1'$ is not weakly proper: its columns are linearly dependent 
for any labeling. Thus, $d_{\rm u}(B_1)=2$. After a proper labeling,  the corresponding minimum distances $d_{\rm q}$ achieve ultimate distances
 $d_{\rm q}(B_2)=d_{\rm u}(B_2)=3$, $ d_{\rm q}(B_3)=d_{\rm u}(B_3)=~3$. 
 \end{example}

The following theorem suggests an efficient algorithm for finding the ultimate distance $d_{\rm u}$ for a given base matrix of the code.

\begin{theorem} \label{thm1}
 For an $r\times n$ base matrix $B$ with $r\leq n$, we have $d_{\rm u}(B)=s+1$, where $s$ is the smallest integer such that $B$  is equivalent to 
 \begin{equation} \label{eq:du}
    B'=\left(
    \begin{array}{c|c}
    A & U \\ \hline
    \bs 0 & V
          \end{array}
            \right),
    \end{equation}
   and $A$ is an $s\times (s+1)$ matrix.
Moreover, $A$ is proper and there exists a $q$-ary labeling ${\mathcal M}$ such that the code with parity-check matrix $H({\mathcal M},B)$ achieves minimum distance $d_{\rm u}(B)$ for any  $q>1+ {n-1\choose s-1}$. 
\end{theorem}
\begin{proof}
If $B$ is equivalent to a matrix $B'$ of the form~(\ref{eq:du}), where $A$ is of size $s\times (s+1)$, then for any labeling ${\mathcal M}$, the first $s+1$ columns of~$H({\mathcal M},B)$ are dependent, hence the corresponding code has minimum distance at most $s+1$. 
Next, we prove achievability 

Now let $s$ be the minimal value. 

\begin{figure}
    \centering
    \includegraphics[width=1.0\linewidth]{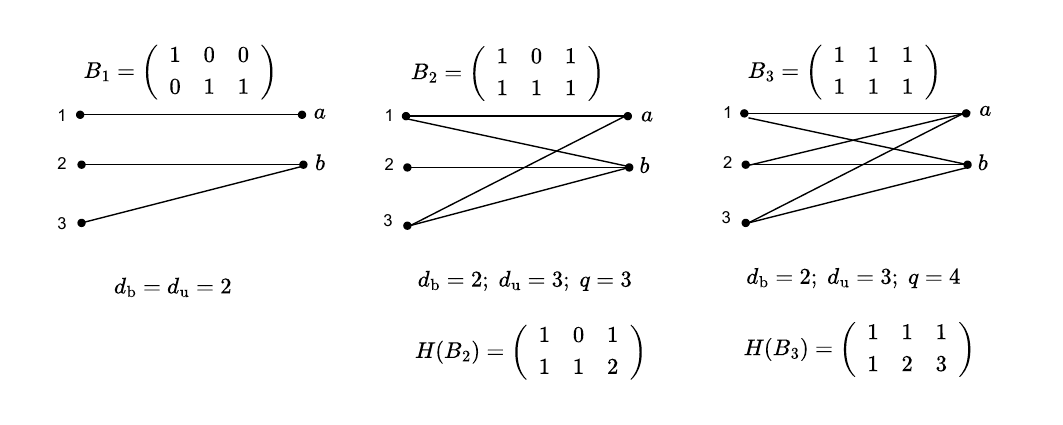}
    \caption{Tanner graphs for proper and improper base matrices}
    \label{fig:Tanner}
\end{figure}

Assume that $A$ is improper, $\mathcal {C}$ and $\mathcal V$ are sets of check and variable nodes of the Tanner corresponding graph. 
It follows from Hall's marriage theorem \cite[Th. 5.1]{LW01} that there exists a set $V\subset \mathcal{V}$ such that $|\Gamma(V)|<|V|$, where $\Gamma(V)\subset \mathcal C$ is the set of nodes adjacent to $V$. 
In Fig.~\ref{fig:Tanner}a, $s=2$, $V=\{2,3\}$, $|V|=2$, $\Gamma(V)=\{b\}$, $|\Gamma(V)|=1$.
The biadjacency matrix $A'$ of the subgraph with $s'=|\Gamma(V)|=1$ nodes from $\mathcal C$ and $t>s'$ nodes of $V$ is proper. The elements in positions $\left\{(i,j)\;|\;i\in \mathcal{C}\setminus \Gamma(V),j\in V\right\}$ are zero. 
The existence of such $A'$ contradicts to the assumption of minimality of $s$. In Fig.~\ref{fig:Tanner}a, $A'=(1,1)$ is biadjacency matrix of the subgraph with $\mathcal V'=\{2,3\}$ and $\mathcal C'=\{b\}$.       



Let $J\subset\{1,...,n\}$, $|J|=s+1$, and $a(J)= s$ (see Def.~\ref{active}). Let $A_J$ be the $s\times(s+1)$ submatrix of active rows in the $r\times (s+1)$ matrix $B_J \subset B$. 
By permuting columns indexed by $J$ and active rows of $A_{J}$ to the first $s+1$ and $a(J)$ positions, respectively, we obtain  $B'$ equivalent to $B$ in the form \eqref{eq:du} with $A=A_J$.  
Thereby, we proved that $A$ in \eqref{eq:du} is proper.

The final step is to estimate a lower bound on the alphabet size required for weakly proper labeling of $B$. 

The first $t=s$ columns of $B$ are labeled as in Corollary~\ref{col1} with success probability $p_q$. Next, for $t=s+1,...,n$, we label a column $t$ such that each $s-1$-subset from the first 
$t-1$ columns together with the $t$-th column is weakly properly labeled. The number of subsets to be checked is at most $N_t=\binom{t-1}{s-1}$. 
The success probability at step $t$ is the probability that the newly labeled column will be successful for all $N_t$ cases. Thus, the failure probability 
\begin{equation} \label{label} 
P_t=1-p_q^{N_t}=1-\left( 1-\frac{1}{q-1} \right)^{N{_t}}
\end{equation}
Let  $q=\ell N_t+1$, $\ell \ge 1$, then
\[
P_t=1-\left( 1-\frac{1}{\ell N_t} \right)^{N{_t}} \xrightarrow[N_t\to\infty]{} 1-e^{-1/\ell} \xrightarrow[\ell\to\infty]{}0.
\]
Thus, for the last step $t=n$, the failure probability $P_n<1$ if $q\ge N_n+1=\binom{n-1}{s-1}+1$, and $P_n$ tends to zero if $q\to \infty$.
\end{proof}

The following notion is useful for applying Theorem~\ref{thm1} to the analysis and construction of NB LDPC codes. 
\begin{definition}\label{StSet}
Given a binary LDPC code of length $n$ with parity-check matrix $H$, a subset $S\subseteq \{1,...,n\}$ is a {\em stopping set\/} if the submatrix $H_I$ contains no rows of weight one. 
 A stopping set {\color{blue} $S$} is called {\em minimal\/} if none of its subsets of smaller size is a stopping set.
\end{definition}

\begin{corollary} \label{col2}
 In the $r \times n$ base matrix $B$, all subsets of column indices $J \subseteq \{1,...n\}$ of size $|J|=d _{\rm u}$, such that the submatrix $A$ of active rows of $B_J$ is proper, are stopping sets.   
\end{corollary}
\begin{proof}
We have to show that $A$ cannot have rows of weight one. W.l.g,  
let the first row of the proper matrix $A$ be of weight one and contain the only nonzero element at the first position. Deleting the first row and the first column, we obtain the proper matrix with 
$d_{\rm u}-1$ columns. It contradicts the assumption that $d_{\rm u}$ is the ultimate distance.
\end{proof}

By Theorem~\ref{thm1} and Corollary~\ref{col2}, $d_{\rm u}$ is the minimum size of such {stopping set $S$ that $a(S)<|S|$. Next, we present an algorithm for finding $d_{\rm u}$ that uses solely the code base matrix~$B$.  

{\bf Algorithm $d_{\rm u}$}:
\begin{enumerate}
    \item Input: Base matrix $B$.
    \item
For each $s=2,3,..,n$, for each $J\in \{1,...,n\}$, $|J|=s$, if $J$ is a stopping set, and if $a(J)<s$, set $d_{\rm u}=s$.  Stop.
\item Return: The ultimate distance $d_{\rm u}$. 
\end{enumerate}

The number of iterations is exponential with $n$, but the iterations are very simple. From \cite{krishnan2007computing} 
follows that the polynomial complexity algorithm for finding stopping sets does not exist. Since we are interested only in short high-rate codes, the search complexity for practical length $n$ is not too large.  

To search for good NB codes, we need to find the minimum distance for different labelings of the same base matrix. In the following algorithm, we search for the minimum distance, assuming that collections $\mathcal{J}_w$ of stopping sets of weight $w \le d_{\rm u}$ for a given base matrix $B$ are known. 
(Algorithm $d_{\rm u}$ exhaustively searches over all subsets of weight up to $w\le d_{\rm u}$ and therefore it can be used for finding all $\mathcal{J}_w$).

{\bf Algorithm} $d_{\rm q}$
\begin{enumerate}
    \item Inputs: $B$, labeling $\mathcal M$, ultimate distance $d_{\rm u}$.
   \item for $w=2,...,d_{\rm u}$,
        \begin{enumerate}
            \item For all $J\in\mathcal{J}_w$, 
            compute the rank $p={\rm rank}(H_{J})$. If $p=w-1$, set $d_{\rm q}=w$. Stop.  
                \end{enumerate}
     \item Return $d_{\rm q}$           
\end{enumerate}
The complexity of verifying the minimum distance for a certain labeling is proportional to $|\mathcal{J}_{d_{\rm u}}|$, which is, in practice, much less than  the number $N_n$ of tested submatrices 
in the proof of Theorem~\ref{thm1}.

The complexity of the search for stopping sets of weight up to $d_{\rm q}$ is upper-bounded by $\binom{n}{d_{\rm q}}$. The complexity of computing the rank of the submatrix over GF($q$) does not exceed $r^3\log^2(q)$. Thus, the complexity exponent of the algorithm is $\log \binom{n}{d_{\rm q}}$, i.e., it
is of the same order as it would be for binary codes.  

Some numerical results are given in Section~\ref{examples}.

\section{Random Coding bounds on $b_{\nu}$}
In this section, we first revisit known upper bounds on the number of incomplete rank submatrices $H_{I}$ for the ensembles of $q$-ary LDPC codes. 
Then the new bound for the Gallager ensemble of $q$-ary LDPC codes is derived.

Most of the results valid for binary LDPC codes can be generalized to $q$-ary LDPC codes. Having an ensemble of parent binary codes, we obtain an ensemble of $q$-ary LDPC codes, replacing the nonzero elements of base parity-check matrices $B$ in the binary ensemble with random nonzero elements of GF($q$).

In \cite{liva2013bounds}, the ensemble of $q$-ary LDPC codes determined by $B$  with i.i.d. entries which take on zero value  with probability $1-p$ and $q-1$ nonzero values with probability $p/(q-1)$ each,  
$p=J/r=K/n$ (ensemble G in \cite{litsyn2002ensembles}) was analyzed.
From \cite[Theorem 1]{liva2013bounds} follows an estimate on $b_{\nu}$:
\begin{equation} \label{liva}
    b_{\nu}\le \frac{\binom{n}{\nu}}{q-1}    
    \sum_{t=1}^\nu (q-1)^t \binom{\nu}{t} 
    \left( \frac{q-1}{q}  \left(1-\frac{pq}{q-1}  \right)^t+\frac{1}{q} \right )^r
\end{equation}
The weight enumerator-based spectra
bound 
for $q$-ary codes gives (see \cite[Theorem 3]{liva2013bounds})
\begin{equation} \label{gen_spectrQ}
b_\nu \le   \binom{n}{\nu} \min \left\{1 , \frac{1}{q-1}
\sum_{w=1}^\nu S_w \frac{\binom{\nu}{w}}{\binom{n}{w}} 
 \right\}.
\end{equation}

To compute the weight enumerators $\{S_w\}$ in \eqref{gen_spectrQ} for $q$-ary LDPC codes whose base matrices belong to the Gallager ensemble of $(J,K)$-regular $q$-ary LDPC codes of length $n$ (ensemble B in \cite{litsyn2002ensembles})
we can use formulas from~\cite[Section~5]{gallager}:
\begin{equation} 
S_w=\left[(q-1)^w\binom{n}{w} \right]^{1-J} \left( \left[g(s)^{r/J}\right]_w \right)^J,
\end{equation}
where $[f(s)]_w$ denotes $w$-th coefficient in the series expansion of $f(s)$, and
\begin{equation}
g(s)=
\frac{1}{q}(1+(q-1)s)^{K}+\frac{q-1}{q}(1-s)^{K}  
\end{equation}
is the weight generating function of the $q$-ary sequences  of
length $n$ satisfying the nonzero part of one $q$-ary parity-check equation.
 
Next, we derive a new bound that is independent of the alphabet size. 
Erasure combinations in $B_{I}$ of weight one can immediately be corrected by the belief propagation (BP) decoder. We exclude such rows of $B_{I}$ from the consideration.  

Erasure combinations covering stopping sets of $B_{I}$ are not correctable by the BP decoder, but the ML decoder can correct some of them. For an erasure combination $I$, the ML decoder solves the system of linear equations \eqref{main1}.
To check whether $I$ is a stopping set, it is enough to consider only the base submatrix $B_I$. Some stopping sets are recoverable by the ML decoder, and some are not. 
The following lemma specifies the cases when $I$ is not recoverable, i.e.  \eqref{main1}  has multiple solutions independently of labeling $B$.

\begin{lemma} \label{lm:UDcondition}
For some proper labeling  $\mathcal M$ of $B$, 
if $B_{I}$ is a proper matrix, then the stopping set $I$ is uniquely correctable by ML decoder if $a(I) \ge |I|$. If $a(I)<|I|$, the unique decoding is impossible regardless of whether  $B$  is proper or not.
\end{lemma}

\begin{proof}


If $a(I)\ge|I|$ we can find among active rows a square submatrix for which there exists a labeling which provides   
rank$(H_I(B_{I}))\ge|I|$ that guarantees unique solution of \eqref{main1}. 

In the case $a(I)<|I|$, the rank of $H_I(B_I)$ is smaller than the number of unknowns. The solution either does not exist or is not unique. Since at least one solution exists (correct codeword), there are multiple solutions of the system \eqref{main1}. 
\end{proof}

\begin{corollary}
\label{improp}
If $B_{I}$ is improper, then the set $I$ is not a minimal stopping set but it includes the smaller stopping set as a subset.        
\end{corollary}
\begin{proof}
The proof uses arguments based on Hall's marriage theorem. See the analysis of improper matrices in the proof of Theorem~\ref{thm1}.
\end{proof}
From  Lemma~\ref{lm:UDcondition} and Corollary \ref{improp} follows that if $a(I)<|I|$ then the set $I$ cannot be corrected by ML decoder, regardless of whether  $B_I$ is proper or not. If $a(I)\ge |I|$ and $B_I$ is improper, then $I$ contains a smaller stopping set.
In what follows, we count only the number of uncorrectable erasure combinations corresponding to the minimal stopping sets. 
Then, we enumerate all erasure patterns which contain these stopping sets. 

Next, we consider an ensemble of random NB LDPC codes  obtained by random labeling base codes from the  Gallager's ensemble of binary $(J,K)$-LDPC codes by random equiprobable nonzero elements of GF$(q)$.   

Recall the definition of the Gallager ensemble of binary LDPC codes. 
\begin{definition}
The parity-check matrices 
of $(J,K)$-regular $[n,k]$-LDPC codes, ($n=MK$, $r=MJ$), 
from the Gallager ensemble consist of $J$ strips, each strip of size $M=r/J$. The first strip is the concatenation of $K$ identity matrices of order $M$. The other $J-1$ strips are obtained as random permutations of the columns of the first strip. 
\end{definition}

\begin{theorem} \label{th:3}
In the Gallager ensemble of $(J,K)$-regular LDPC codes of length $n=KM$, there exist codes over GF($q$) such that for sufficiently large $q$ the number of incomplete rank submatrices of $\nu$-columns is
\begin{equation}\label{eq:upper}
 b_\nu\le\sum_{w=2}^{\nu} S_w \binom{n-w}{\nu-w},
\end{equation}    
where 
\begin{equation} \label{eq:T3}
  S_w=\binom{n}{w}^{1-J}\;\;\sum_{a=1}^{w-1}
\left[
\left( 
\Big[\psi_{1}(\lambda,s)^M -1\Big]_{s^w} 
\right)^J
\right]_{\lambda^a}.
\end{equation}
and $\psi_{1}(\lambda,s)$ is defined in \eqref{eq:psi} below.
\end{theorem}
In \eqref{eq:T3}, $\Big[p(x_1,...,x_n)\Big]_{x_i^j}$ denotes the coefficient of $x_i^j$, $i=1,..,n$ of the muti-variable polynomial $p(x_1,...,x_n)$.  
\begin{proof} 
According to Lemma~\ref{lm:UDcondition} and Corollary \ref{improp}, for a given number of erasures $\nu=|I|$, we can estimate the number of incomplete rank submatrices $H_{I}$ by enumerating the submatrices with different number of active rows $a(I)$, $a(I)<\nu$ corresponding to stopping set $I$ in the base matrix $B_I$. The set $\mathcal B_\nu$ of uncorrectable weight $\nu$ erasure patterns includes the subset $\mathcal S_\nu \subseteq \mathcal B_\nu $ consisting of minimal stopping sets.


We start with estimating $S_\nu=|\mathcal S_\nu|$. The derived bound will later be used to estimate $b_\nu=|\mathcal B_\nu | $.


For counting the number of stopping sets $I$ of size $\nu$, such that the corresponding submatrix $B_I$ has $a(I)<\nu$ active rows, 
we introduce a two-variable generating function of the number of active rows $a$ of weight two or more and the weight $w$ of erasure combination $I$ in the matrix with $r$ rows as 
\[
\psi_r(\lambda,s)=\sum_{w\in\{0,2,3,...n\}}\sum_{a=0}^{r} \psi_{a,w}\lambda^as^w,
\]
where coefficient before $\lambda^as^w$ is the number of weight $w$ erasure combinations producing $a$ rows of weight at least 2.
For a single row, taking into account that erasure combinations of weight zero and one do not produce a row of weight 2 or more, and $a \in \{0,1\}$, we obtain  
\begin{eqnarray}\label{eq:psi}
\psi_1(\lambda,s)&=&
\lambda^0+\lambda \left[\binom{K}{2}s^2+...+\binom{K}{K}s^K\right] 
\nonumber\\ 
&=& 1+\lambda\left[(1+s)^K-1-Ks \right].
\end{eqnarray}

The nonzero positions of rows of one strip do not intersect. Thus, the two-variable generating function  $\eta(\lambda,s)$ for one width $M$ strip in the base matrix from the Gallager ensemble is equal to 
\[
\eta(\lambda,s)=\sum_{w=2}^{n} \eta_w(\lambda)s^w=\psi_1(\lambda,s)^M-1.
\]
For a fixed $w$, the generating function of the number of active rows is 
\[
\eta_w(\lambda) = \Big[ \psi_1(\lambda,s)^M-1 \Big]_{s^w},
\]
where $\left[ \psi(\lambda,s)\right]_{s^w}$ denotes the coefficient for $s^w$ in the series expansion of $\psi(\lambda,s)$.

Assuming that all erasure sequences of weight $w$ are equiprobable, for a given $w$, 
probabilistic moment generating function for $a$ is 
\begin{equation}\label{tilde}
\eta_w(\lambda)/\binom{n}{w} =\binom{n}{w}^{-1}
\Big[ \psi_1(\lambda,s)^M-1 \Big]_{s^w}.
\end{equation}

Assuming that the $J$ strips are independent, the 
generating function for the number of active rows in the entire matrix $B$ for the fixed $w$ is  
$\eta_w(\lambda)^J$. The average generating function of the number of active rows 
among $\binom{n}{w}$ sequences of weight $w$
is
\begin{equation*}\label{xi}
\xi_w(\lambda)\!=\!\binom{n}{w} \!
\left(\!\frac{\eta_w(\lambda)}{\binom{n}{w}}\!\right)\!^J\!
=\!\binom{n}{w}^{1-J}\!\left(\Big[ \psi_1(\lambda,s)^M\!-\!1 \right]_{s^w}\!\! \Big)^J.
\end{equation*}
Finally, the probability of erasure combinations $I$ for which $a(I)<\nu$ is equal to the sum of the first $\nu-1$ coefficients in the series expansion of $\xi_\nu(\lambda)$.
Thus, we have the formula for $S_\nu$
 
\begin{eqnarray}
S_\nu&=&  
\sum_{a=2}^{\nu-1} \Big[\xi_\nu(\lambda) \Big]_{\lambda^a} \nonumber \\
&=&\binom{n}{\nu}^{1-J}\sum_{a=1}^{\nu-1}
\left[
\left( 
\Big[\psi_1(\lambda,s)^M -1\Big]_{s^\nu} 
\right)^J
\right]_{\lambda^a}.
\end{eqnarray}

To estimate $b_\nu$, notice that according to Corollary~\ref{improp} any uncorrectable erasure pattern $I\in \mathcal B_\nu$ contains as a subset some stopping set from $\mathcal S_\nu$. By summing up all  patterns overlapping with stopping sets, we obtain the following upper bound
\begin{equation}\label{eq:upper2}
 b_\nu\le\sum_{w=2}^{\nu} S_w \binom{n-w}{\nu-w}.
\end{equation}
\end{proof}

\begin{figure}
    \centering
    \includegraphics[width=0.82\linewidth]{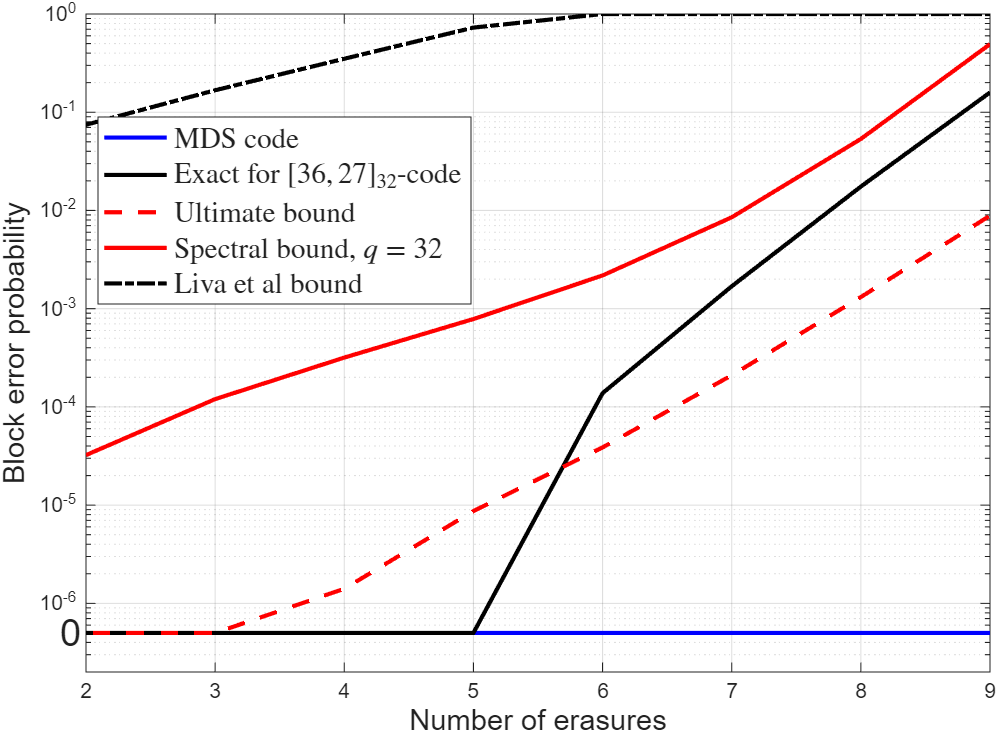}
        \caption{
 Conditional block error probability $P_{\text {block}}(\nu)$ as a function of erasure number $\nu$ for the NB LDPC code $[36,27]_{32}$  from Table~\ref{tab:codes} and randomly labeled Gallager's code with the same parameters. The exact value is obtained by exhaustive search over all erasure combinations. The  new bound \eqref{eq:upper}--\eqref{eq:T3}.  The spectral bound, and Liva et al. bound are defined by formulas \eqref{gen_spectrQ}, and \eqref{liva}, respectively    \label{fig:qec}}
\end{figure}
Comparison of the bounds in Fig. \ref{fig:qec} shows that although NB LDPC codes are not MDS codes, they can provide erasure decoding performance which approaches the performance of MDS codes when their alphabet size is not limited. In the next section, we present examples of short NB LDPC codes with quasi-cyclic (QC) base binary LDPC codes.         
\section{Examples of short NB LDPC codes } \label{examples}

\begin{example} \label{ex:ex1}
Let column weight be $J=2$, and the base matrix be
{\small
\[
B=
\begin{pmatrix}
 1& 0& 0& 0& 1& 1& 0& 0& 1& 0\\ 
 1& 1& 0& 0& 0& 0& 1& 0& 0& 1\\ 
 0& 1& 1& 0& 0& 1& 0& 1& 0& 0\\ 
 0& 0& 1& 1& 0& 0& 1& 0& 1& 0\\ 
 0& 0& 0& 1& 1& 0& 0& 1& 0& 1
 \end{pmatrix} 
\]}
For this (2,4)-regular LDPC code, $g=6$, $d_{\rm b}=3$. Taking rows with numbers 1...4 and columns corresponding to zeros in the last row, we obtain 
{\small
\[
B_1= \left(
\begin{array}{ccccc|ccccc}
 1& 0& 0& 1& 0& 1& 0& 1& 0& 0\\ 
 1& 1& 0& 0& 1& 0& 0& 0& 0& 1\\ 
 0& 1& 1& 1& 0& 0& 0& 0& 1& 0\\ 
 0& 0& 1& 0& 1& 1& 1& 0& 0& 0\\  \hline
 0& 0& 0& 0& 0& 0& 1& 1& 1& 1
 \end{array}
 \right)
\]}
a submatrix of size $4\times 5$ with rank at most 4 independently of labeling. Thus $d_{\rm u}=5$, slightly less than for MDS code $r+1=6$. Labeled matrix over GF(8) 
{\small
\[
H=
\begin{pmatrix}
  1& 0& 0& 0& 5& 1& 0& 0& 1& 0\\ 
  5& 1& 0& 0& 0& 0& 1& 0& 0& 1\\ 
  0& 5& 1& 0& 0& 1& 0& 1& 0& 0\\ 
  0& 0& 5& 1& 0& 0& 2& 0& 3& 0\\ 
  0& 0& 0& 5& 1& 0& 0& 1& 0& 7
\end{pmatrix}
\]}
attains this bound, $d_{\rm q}=d_{\rm u}=5$.

\end{example}
\begin{example} \label{ex:ex2}
Codes with $J=2$ based on complete graphs. 
For any odd $r$, the biadjacency matrix of complete graph $K_r$ of size $r \times n$, $n=\binom{r}{2}$  can be presented in the form 
\begin{equation} \label{eq:qc}
    B=\begin{pmatrix}
     C_1 & C_2& ... & C_{(r-1)/2}   
    \end{pmatrix}
\end{equation}
where $C_i$ are circulant matrices of order $r$. The first row of $C_i$ contains two ones, at positions 1 and  
$i+1$.   All such codes have exactly the same parameters as code from Example~\ref{ex:ex1}: $g=6$, $d_{\rm b}=3$, $d_{\rm u}=5$. 

In particular, for $r=9$ we obtain $[36,27,3]_2$-code from which after proper labeling the $[36,27,5]_q$-code of rate $R=3/4$ may be obtained. 
\end{example}

\begin{example} \label{ex:ex3}
Codes based on complete bipartite graphs $K_{r/2,r/2}$. Short and efficient LDPC codes with $g=8$ can be obtained from bipartite graphs. 

The base matrix of the code is
\begin{equation} \label{eq:complete}
B=\begin{pmatrix}
I&I&...&I\\
I& \mathcal J&...&\mathcal J^{r/2-1} 
\end{pmatrix},    
\end{equation}
where $\mathcal J$ is the cyclic permutation matrix. The parameters of the base code are $[n=r^2/4,n-r,4]_2$. 
For $r=16$, we obtain [64,16,4] base code. The ultimate distance for this code is $d_{\rm u}=6$.
\end{example}

\begin{example} $q$-ary codes from binary QC-LDPC codes with column weight $J=3$.
For binary codes with column weight $J=2$, the binary minimum distance $d_{\rm  b}$ is equal to half of the girth of the Tanner graph. To overcome this limitation, we use form \eqref{eq:qc} to construct more code examples by substituting circulant matrices of column weight $J=3$. 
Examples of codes are presented in Table~\ref{tab:codes}.

\end{example}

In Table \ref{tab:codes}, we give generators of QC base matrices of short rate $R=3/4$ NB LDPC codes whose minimum distance $d_{\rm q}$ achieves the upper bound $d_{\rm u}$. Each base matrix consists of four $r\times r$ circulants, each of which is obtained as $r$ shifts of the corresponding generator.

In the same table, we present the minimum distance of the base LDPC code $d_{\rm b}$, the girth of the base Tanner graph $g$, and the logarithm of the alphabet size $q$.  

Notice that the alphabet size needed for achieving the ultimate distance is much smaller than the estimate in Theorem~\ref{thm1}, and even less than for RS codes of the same length. 

\begin{table}[h]
\renewcommand{\arraystretch}{1.3}
    \centering
      \caption{Generators of QC base matrices of NB LDPC codes over GF$(q)$ of rate $R=3/4$ \label{tab:codes} }
     \begin{tabular}{|c|p{2.75cm}|c|c|c|c|c|}
    \hline
$[n,k], J$    
       & 
Generators  & 
$g_{\rm b}$   & 
$d_{\rm b}$  & 
$d_{\rm u}$ &
$\log_2q$ \\ \hline
[36,27], 2 & (1,2), (1,3), (1,4), (1,5)          & 6 & 3 &  5& 4\\ \hline 
[36,27], 3 & (1,2,4), (1,2,5), (1,2,7), (1,3,6)  & 4 & 4 &  6& 5\\ \hline
[52,39], 3 & (1,2,4), (1,2,6), (1,3,8), (1,4,8)  & 4 & 4 &  8& 7 \\ \hline 
[76,57], 3 & (1,2,5), (1,3,10), (1,6,12), (1,2,4)& 4 & 4 &  9& 6 \\ \hline 
\end{tabular}
\end{table}

\begin{figure}
    \centering
       \includegraphics[width=0.9\linewidth]{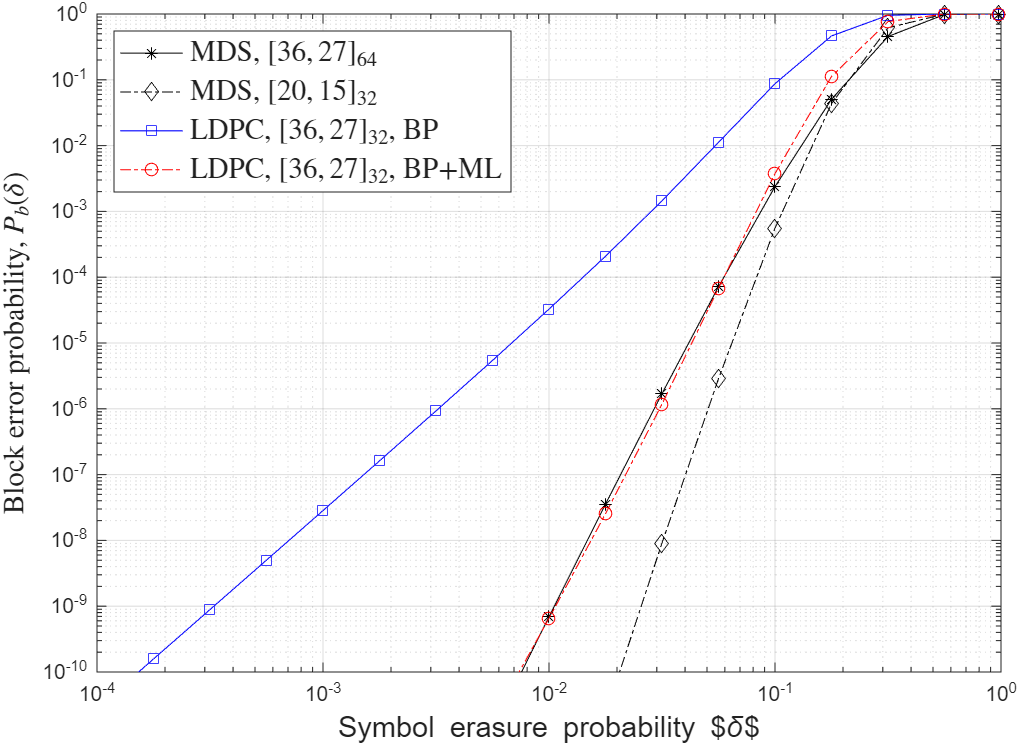}
    \caption{Numerically computed block error probability in QEC for LDPC code and two examples of MDS codes}
    \label{fig:ErrProb}
\end{figure}

\begin{table}
    \centering
    \renewcommand{\arraystretch}{1.3}
    \renewcommand{\tabcolsep}{0.1cm}
     \caption{Average number of arithmetic operations for decoding LDPC and RS codes in QEC with erasure probability $\delta=0.02$ }
    \begin{tabular}{|c|c|c|c|c|c|}
    \hline
     Code    &  Addi\-tions & \parbox{2cm}{Multipli\-cations}&Inversion &  Total& $P_{\rm b}(\delta)$\\ \hline
 $[36,27,6]_{32}$   & 108  & 110 & 1.72 & 220  & $1.0\cdot 10^{-6}$ \\ \hline
 $[20,15,6]_{32}$  & 303   & 442 & 30  & 775  & $1.0\cdot 10^{-6}$\\ \hline
 $[36,27,10]_{64}$   & 1098  & 1578 & 72 & 2748 & $1.0\cdot 10^{-8}$\\ \hline
    \end{tabular}
    \label{tab:Comp}
\end{table}

The decoding error probabilities in QEC  for MDS codes and for LDPC codes, both of rate $R=3/4$, are shown in Fig.~\ref{fig:ErrProb}. 
The LDPC code is the code from the second line in Table~\ref{tab:codes}. As competitors, we consider two MDS codes, $[20,15,6]_{32}$ and $[36,27,10]_{32}$.
LDPC code under BP decoding is less reliable than MDS codes. However, the same code can be decoded using a hybrid decoder that combines a BP decoder with the solution of a reduced-order linear system (ML decoder of reduced complexity). In this case, the error probability is the same as for $[20,15,6]_{32}$ MDS code. 

The average numbers of arithmetic operations over GF$(q)$  for these three scenarios are summarized in Table~\ref{tab:Comp}. For LDPC codes, the hybrid (BP+ML) decoder was analyzed. For decoding RS codes, Forney's algorithm was used. It can be seen that the decoder for the LDPC code is much simpler than that for RS codes.

\section{Conclusion}
In this work, we suggest an approach to the analysis and design of short NB LDPC codes by assigning field elements to nonbinary entries of a binary base matrix. First, we introduce the notion of ultimate distance achievable for a given base matrix regardless of the alphabet size $q$. Next, we propose an algorithm for finding the ultimate distance. A low-complexity  algorithm for finding the minimum distance of the NB LDPC code is presented. It is based  on the analysis of the code base matrix. We also use the random coding technique for estimating the erasure-correcting capability of NB LDPC codes. 

We demonstrate that NB LDPC codes are strong competitors of MDS codes for some practical scenarios.  

Many questions related to this approach remain unanswered. For example, the estimate on the alphabet size $q$ should be tightened.   


\section*{Acknowledgment}




\bibliographystyle{IEEEtran}
\bibliography{ref,listdec_bec,coding,extra}

\end{document}